\documentclass[11pt]{article}

\usepackage{xcolor}
\usepackage[colorlinks=true,citecolor=blue,linkcolor=blue]{hyperref}
\hypersetup{
    colorlinks,
    linkcolor={red!50!black},
    citecolor={blue!50!black},
    urlcolor={blue!80!black}
}

\usepackage[]{amsmath,amssymb,amsfonts,latexsym,amsthm,enumerate,fullpage,xcolor,cite,bbm}
\usepackage[nameinlink, noabbrev, capitalize]{cleveref}
\crefname{prop}{Proposition}{Propositions}
\crefname{ineq}{inequality}{inequalities}
\creflabelformat{ineq}{#2(#1)#3}

\def\email#1{{\tt#1}}

\newtheorem{theorem}{Theorem}

\newtheorem{proposition}{Proposition}

\newtheorem{definition}{Definition}
\newtheorem{remark}{Remark}

\crefname{THM}{Theorem}{Theorems}

\usepackage{subcaption}
\usepackage{graphicx}
\usepackage{tikz}
\usetikzlibrary{positioning}
\usetikzlibrary{decorations.pathmorphing}
\usetikzlibrary{automata,positioning}
\usetikzlibrary{arrows}
\usetikzlibrary{arrows.meta}

\newcommand{\N}{\mathbb{N}}
\newcommand{\Z}{\mathbb{Z}}

\DeclareMathOperator*{\argmin}{\arg\!\min}
\DeclareMathOperator*{\argmax}{\arg\!\max}

\usepackage{algorithm,algpseudocode}

\algblockdefx[Loop]{Loop}{EndLoop}[1][]{\textbf{Loop} #1}{\textbf{End Loop}}

\algrenewcomment[1]{\(\triangleright\) #1}
\algnewcommand{\LineComment}[1]{\State \(\triangleright\) #1}
\usepackage{tabto} 

\begin{document}
\title{On the Approximability of Time Disjoint Walks}
\author{
Alexandre Bayen\thanks{University of California, Berkeley. \email{\{bayen,evinitsky\}@berkeley.edu}.}
\and Jesse Goodman\thanks{Cornell University. \email{jpmgoodman@cs.cornell.edu}.}
\and Eugene Vinitsky\footnotemark[1]}
\maketitle

\begin{abstract}
We introduce the combinatorial optimization problem Time Disjoint Walks (TDW), which has applications in collision-free routing of discrete objects (e.g., autonomous vehicles) over a network. This problem takes as input a digraph \(G\) with positive integer arc lengths, and \(k\) pairs of vertices that each represent a \emph{trip demand} from a source to a destination. The goal is to find a walk and delay for each demand so that no two trips occupy the same vertex \emph{at the same time}, and so that a min-max or min-sum objective over the trip durations is realized.

We focus here on the min-sum variant of Time Disjoint Walks, although most of our results carry over to the min-max case. We restrict our study to various subclasses of DAGs, and observe that there is a sharp complexity boundary between Time Disjoint Walks on oriented stars and on oriented stars with the central vertex replaced by a path. In particular, we present a poly-time algorithm for min-sum and min-max TDW on the former, but show that min-sum TDW on the latter is NP-hard.

Our main hardness result is that for DAGs with max degree \(\Delta\leq3\), min-sum Time Disjoint Walks is APX-hard. We present a natural approximation algorithm for the same class, and provide a tight analysis. In particular, we prove that it achieves an approximation ratio of \(\Theta(k/\log k)\) on bounded-degree DAGs, and \(\Theta(k)\) on DAGs and bounded-degree digraphs.
\end{abstract}


\section{Introduction}
\subsection{Related work}
Disjoint Paths is a classic problem in combinatorial optimization that asks: given an undirected graph \(G\), and \(k\) pairs of vertices, do there exist vertex-disjoint paths that connect each pair? This problem captures the general notion of \emph{connection without interference}, and has subsequently received much attention due to its applicability in areas like VLSI design~\cite{VLSI2,VLSI1} and communication networks~\cite{communication2,communication1}.

These applications have motivated many variants of this basic problem. For example, one may choose the underlying graph to be undirected or directed, and the disjointness constraint to be over vertices or edges. As an optimization problem, one may consider the \emph{maximum number of pairs} that can be connected with disjoint paths, the \emph{minimum number of rounds} necessary to connect all pairs (where all paths in a round must be disjoint)~\cite{kleinberg}, or the \emph{shortest set of disjoint paths} to connect all pairs (if all pairs can, in fact, be disjointly connected)~\cite{kobayashi}.

A few flavors of Disjoint Paths are tractable: for example, if \(k\) is fixed or \(G\) has bounded tree-width, then there exists a poly-time algorithm for finding vertex-disjoint paths on undirected graphs~\cite{treewidth,seymour}. Many interesting variants of Disjoint Paths are, however, extremely difficult. Indeed, finding vertex-disjoint paths on undirected graphs is one of Karp's NP-complete problems~\cite{karp}. Furthermore, nearly-tight hardness results are known for finding the maximum set of edge-disjoint paths in a directed graph with \(m\) edges: there exists an \(O(\sqrt{m})\)-approximation algorithm~\cite{kleinberg}, and it is NP-hard to approximate within a factor of \(m^{1/2-\epsilon}\), for any \(\epsilon>0\)~\cite{tight-hardness}. For detailed surveys on the complexity landscape of Disjoint Paths variants, see~\cite{kleinberg,kobayashi}.

\subsection{Contributions}
Despite the great variety of Disjoint Paths problems that have been considered in the literature, it appears that little attention has been given to variants that relax the \emph{disjointness} constraint, even though many natural applications do not always require paths to be completely disjoint. Consider, for example, the application of safely routing a collection of fully autonomous (and obedient) vehicles through an otherwise empty road network. In such a situation, we can certainly prevent collisions by routing all vehicles on disjoint paths. However, it is not difficult to see that if we have full control over the vehicles, using disjoint paths is rarely necessary (and, in fact, can be highly suboptimal).

Applications of this flavor motivate a new variant of Disjoint Paths, which roughly asks: given a graph \(G\) and \(k\) pairs of vertices that each represent a \emph{trip demand}, how should we assign a delay and a path to each trip so that (1) trips are completed as quickly as possible, and (2) no two trips \emph{collide} (i.e., occupy the same location at the same time). While there are problems in the literature (that do not wield the name ``Disjoint Paths'') that seemingly come close to capturing this goal, they exhibit some key differences. In particular, multicommodity flows over time~\cite{MFOT1,MFOT2} and job shop scheduling~\cite{JSS1} seem, at first glance, very related to our problem. However, the former does not support ``atomic'' commodities (i.e., whose demand-satisfying flow must be unsplittable over the network and time), and the latter does not capture the flexibility of scheduling job operations over any appropriate walk in a network.

As such, we are motivated to formalize and study this new variant of Disjoint Paths that relaxes the classical disjointness constraint to a ``time disjointness'' constraint. In particular, our contributions are fourfold:
\begin{itemize}
\item We introduce a natural variant of Disjoint Paths, which we call \emph{Time Disjoint Walks} (TDW). To the best of our knowledge, this is the first simple model that captures the notion of collision-free routing of \emph{discrete objects} over a shared network.

\item We give an intuitive sense for precisely \emph{when} Time Disjoint Walks becomes hard, by showing that it can be exactly solved (under min-sum and min-max objectives) in poly-time on oriented stars, but is NP-hard (under the min-sum objective) on a slightly more expressive class.

\item We prove that min-sum Time Disjoint Walks is APX-hard, by providing an L-reduction from a variant of SAT. In fact, our reduction shows that this result holds even for directed acyclic graphs (DAGs) of max degree three (\(\Delta\leq 3\)).

\item We describe an intuitive approximation algorithm for min-sum Time Disjoint Walks, and provide a tight analysis: we show that it achieves an approximation ratio of \(\Theta(k/\log k)\) on bounded-degree DAGs, and \(\Theta(k)\) on DAGs and bounded-degree digraphs.

\end{itemize}

We formally introduce Time Disjoint Walks in \cref{section:TDW}. In \cref{section:boundary}, we present an exact poly-time algorithm for TDW on oriented stars, and show that min-sum TDW is NP-hard on a slightly more complex class. In \cref{section:approx-prelims} we provide some useful definitions regarding approximation. In \cref{section:hardness} we prove our APX-hardness result. In \cref{section:approx-alg} we describe our approximation algorithm, and provide bounds on its performance. In \cref{section:conclusions} we state our conclusions and present some open problems. A preliminary version of this paper appeared in~\cite{goodman}.

\section{Time Disjoint Walks}\label{section:TDW}
\begin{figure}[t]
\centering
\begin{tikzpicture}[line width=1pt]
\tikzset{vertex/.style = {shape=circle,draw,minimum size=0.5em, inner sep=0pt}}
\tikzset{edge/.style = {->,> = latex'}}
\node[vertex,label=left:\(s_1\)] (1c1) at (0,0) {};
\node[vertex] (1c2) at (1,0) {};
\node[vertex] (1c3) at (2,0) {};
\node[vertex] (1c4) at (3,0) {};
\node[vertex] (1c5) at (4,0) {};
\node[vertex] (1c6) at (5,0) {};
\node[vertex] (1c7) at (6,0) {};
\node[vertex] (1c8) at (7,0) {};
\node[vertex,label=right:\(t_1\)] (1c9) at (8,0) {};

\node[vertex,label=above:\(s_2\)] (1t2) at (1,1) {};
\node[vertex,label=below:\(t_2\)] (1b2) at (1,-1) {};
\node[vertex,label=above:\(s_3\)] (1t3) at (2,1) {};
\node[vertex,label=below:\(t_3\)] (1b3) at (2,-1) {};
\node[vertex,label=above:\(s_4\)] (1t4) at (3,1) {};
\node[vertex,label=below:\(t_4\)] (1b4) at (3,-1) {};

\node[vertex] (1cyc) at (4,-1) {};

\node[vertex,label=above:\(s_5\)] (1t6) at (5,1) {};
\node[vertex,label=below:\(t_5\)] (1b6) at (5,-1) {};
\node[vertex,label=above:\(s_6\)] (1t7) at (6,1) {};
\node[vertex,label=below:\(t_6\)] (1b7) at (6,-1) {};
\node[vertex,label=above:\(s_7\)] (1t8) at (7,1) {};
\node[vertex,label=below:\(t_7\)] (1b8) at (7,-1) {};

\draw[edge] (1c1) to (1c2);
\draw[edge] (1c2) to (1c3);
\draw[edge] (1c3) to (1c4);
\draw[edge] (1c4) to (1c5);
\draw[edge] (1c5) to (1c6);
\draw[edge] (1c6) to (1c7);
\draw[edge] (1c7) to (1c8);
\draw[edge] (1c8) to (1c9);

\draw[edge] (1t2) to node[right,midway] {\(2\)} (1c2);
\draw[edge] (1t3) to node[right,midway] {\(3\)} (1c3);
\draw[edge] (1t4) to node[right,midway] {\(5\)} (1c4);
\draw[edge] (1c2) to (1b2);
\draw[edge] (1c3) to (1b3);
\draw[edge] (1c4) to (1b4);

\draw[edge] (1t6) to node[left,midway] {\(5\)} (1c6);
\draw[edge] (1t7) to node[left,midway] {\(6\)} (1c7);
\draw[edge] (1t8) to node[left,midway] {\(7\)} (1c8);

\draw[edge] (1c6) to (1b6);
\draw[edge] (1c7) to (1b7);
\draw[edge] (1c8) to (1b8);

\draw[edge,bend left=45] (1c5) to (1cyc);
\draw[edge,bend left=45] (1cyc) to (1c5);
\draw[edge,bend left=90] (1c1) to node[above,midway] {\(11\)} (1c9);

\node[] at (4,-2) {(i)};

\node[vertex,label=above:\(s_1\)] (2tl) at (9,1) {};
\node[vertex,label=below:\(s_2\)] (2bl) at (9,-1) {};
\node[vertex] (2c) at (10,0) {};
\node[vertex,label=above:\(t_2\)] (2tr) at (11,1) {};
\node[vertex,label=below:\(t_1\)] (2br) at (11,-1) {};

\draw[edge] (2tl) to (2c);
\draw[edge] (2bl) to node[below,midway] {\(2\)} (2c);
\draw[edge] (2c) to (2br);
\draw[edge] (2c) to (2tr);

\node[] at (10,-2) {(ii)};
\end{tikzpicture}
\caption{(Unlabeled arcs have length 1): (i) A TDW instance with an optimal solution that contains cycles and intersecting walks, even though disjoint paths exist. (ii) A TDW instance with an obvious optimal solution, or a Shortest Disjoint Paths instance with no solution.}
\label{fig:tdw-cycles}
\end{figure}
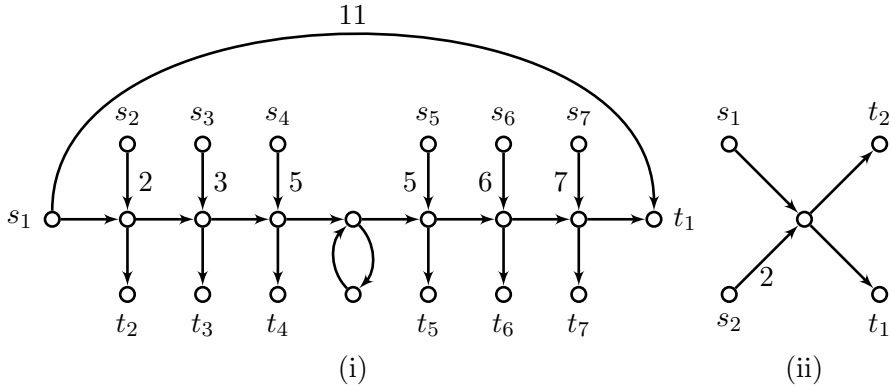
We must first mention a few preliminaries: given \(a, b\in\Z\), define \([a,b]:=\{x\in\Z\mid a\leq x \leq b\}\), and for \(b\in\Z\), we write \([b]:=[1,b]\). Note that for \(b<1\), \([b]=\emptyset\). We assume directed graphs (digraphs) have no parallel arcs in the same direction, and for a digraph \(G\), we let \(\Delta\) denote its max \emph{total} degree (indegree \emph{and} outdegree contribute to a vertex's total degree). Given a digraph \(G:=(V,E)\), and \(u,v\in V\), we define a \emph{walk} \(W\) from \(u\) to \(v\) in \(G\) as a tuple \((w_1,w_2,\dots,w_l)\) of vertices such that \(w_1=u, w_l=v\), and \((w_i,w_{i+1})\in E\) for each \(i \in [l-1]\). Note that a vertex may be repeated.
Given a digraph \(G\) with arc lengths \(\lambda:E\rightarrow\Z_{\geq1}\), and a walk \(W=(w_1,w_2,\dots,w_l)\) in \(G\), we let \(|W|:= l\) denote the \emph{cardinality} of the walk, and we define for every \(j\in[l]\) the \emph{length of the walk up to its \(j^{th}\) vertex} as
\begin{align*}
\lambda(W,j) := \sum_{i\in[j-1]}\lambda(w_i,w_{i+1}).
\end{align*}
For convenience, we let \(\lambda(W) := \lambda(W,l)\) denote the total length of the walk, and if \(W\) is a \emph{path} (i.e., no repeated vertices), we use \(\lambda(W,j)\) and \(\lambda(W,w_j)\) interchangeably. 

Finally, given delays \(d_1,d_2\in\Z_{\geq0}\) and walks \(W_1,W_2\) in \(G\), we say that \((d_1,W_1)\) and \((d_2,W_2)\) are \emph{time disjoint} if, intuitively, a small object traversing \(W_1\) at constant speed after waiting \(d_1\) units of time does not collide/interfere with a small object traversing \(W_2\) at the same speed after waiting \(d_2\) units of time. We consider walks that have not departed, and walks that have already ended, to no longer exist on the network (and thereby not occupy any vertices). Formally, we have: for every \(j_1\in[|W_1|],j_2\in[|W_2|]\) such that the \(j_1^\text{th}\) vertex of \(W_1\) is equal to the \(j_2^\text{th}\) vertex of \(W_2\),
\begin{align*}
d_1 + \lambda(W_1,j_1)\neq d_2+\lambda(W_2,j_2).
\end{align*}

We are now ready to formally define the problem examined in this paper:

\begin{definition}[Time Disjoint Walks (TDW)] Let \(G:=(V,E)\) be a digraph, let \(\lambda:E\rightarrow\mathbb{Z}_{\geq1}\) define arc lengths, and let \(\mathcal{T}:=\{(s_1,t_1),(s_2,t_2),\dots,(s_k,t_k)\}\subseteq V^2\) define a set of demands across unique vertices. For each \(i\in[k]\), find a delay \(d_i\in\mathbb{Z}_{\geq0}\) and walk \(W_i\) from \(s_i\) to \(t_i\) such that the tuples in \(\{(d_i,W_i)\mid i\in[k]\}\) are pairwise time disjoint, and such that \(\sum_{i\in[k]}(d_i+\lambda(W_i))\) or \(\max_{i\in[k]}(d_i+\lambda(W_i))\) is minimized.\footnote{We assume \(G\) contains, for each \((s_i,t_i)\), at least one walk connecting the two vertices.} We refer to the former as \emph{min-sum Time Disjoint Walks}, and the latter as \emph{min-max Time Disjoint Walks}.
\end{definition}

We note that one can construct analogous problems by considering undirected graphs as input, edge lengths and delays that are real-valued, or a definition of \emph{time disjoint} that requires large gaps between arrival times at common vertices (whereas the definition above simply requires a nonzero gap). We leave these variants to future work, noting that our primary goal in this paper is to study a basic flavor of this new combinatorial problem.

Furthermore, our selection of this variant is well-motivated by our original application of routing a collection of identical autonomous vehicles over an empty road network (which, for the sake of this futuristic application, we may assume was built specifically for these vehicles). In particular, we may (1) model the road network as a directed graph, (2) assume that all routed vehicles traverse their walk at the same constant velocity, (3) measure road lengths as the time necessary to traverse it at that velocity, and (4) assume that road lengths are integer multiples of the time length of each vehicle. Additionally, we may motivate our focus on the \emph{min-sum} objective by the desire to find a socially optimal solution.

Finally, we emphasize the novelty of our \emph{time disjoint} constraint by comparing it to the standard \emph{disjoint} constraint used in classical variants of Disjoint Paths. In particular, observe that if we modify the definition of Time Disjoint Walks to use the latter constraint instead of the former, we arrive at the Shortest Disjoint Paths problem~\cite{kobayashi}. However, this constraint makes all the difference: given an instance of Time Disjoint Walks, it is often the case that a solution under the standard \emph{disjoint} constraint is suboptimal if examined under the \emph{time disjoint} constraint. Indeed, the optimal solution under the latter constraint may even include paths that repeat vertices - hence the name Time Disjoint \emph{Walks}; see \cref{fig:tdw-cycles}, part (i). On the other hand, it is easy to construct an instance of Shortest Disjoint Paths that admits an obvious optimal solution under the time disjoint constraint, but does not yield any solution at all under the classical disjoint constraint; see \cref{fig:tdw-cycles}, part (ii).

These observations strongly suggest that there is no simple reduction, in either direction, between Time Disjoint Walks and Disjoint Paths. Furthermore, using time-expanded networks~\cite{MFOT1} to reduce Time Disjoint Walks into Disjoint Paths appears to offer little hope: such reductions will approximately square the size of the original graph, and many variants of Disjoint Paths are hard to approximate within \(m^{1/2-\epsilon}\), for any \(\epsilon>0\)~\cite{tight-hardness}. Thus, an approximation algorithm for Disjoint Paths, applied to a transformed Time Disjoint Walks instance, would likely fail to perform better than a trivial approximation algorithm for Time Disjoint Walks. These observations highlight the novelty of our problem and results.

\section{A complexity boundary}\label{section:boundary}
In order to get an idea of the difficulty of Time Disjoint Walks, we show that min-sum TDW exhibits a sharp complexity boundary, in the sense that it is poly-time solvable on one subclass of DAGs, but becomes NP-hard if we make this class slightly more expressive. In particular, we consider the following two subclasses of DAGs:
\begin{definition}[Oriented stars] A digraph \(G\) is an \emph{oriented star} if it is composed of a central vertex with any number of leaves attached by directed edges.
\end{definition}
Equivalently, an oriented star on \(n\) vertices is simply the complete bipartite graph \(K_{1,n-1}\) with some orientation applied to the edges.
\begin{definition}[Runways]
A digraph \(G\) is a \emph{runway} if it is composed of a central directed path with any number of leaves attached (anywhere along the path) by directed edges.
\end{definition}
One may think of a runway as an oriented star with the central vertex replaced by a directed path, or as a caterpillar graph with an orientation applied to the edges so as to produce a central directed path.

\subsection{Time Disjoint Walks on oriented stars in polynomial time}
We present poly-time \cref{algs:star} for Time Disjoint Walks on oriented stars, and show that it in fact optimally solves both min-sum and min-max TDW. The main idea in this algorithm and its analysis is that each trip \(i\in[k]\) is allocated a unique \emph{time slot} during which it is ``allowed'' to utilize the central vertex of the oriented star (i.e., the only vertex it will share with other paths). We often toggle between reasoning about the delay assigned to a trip and the time slot (on the central vertex) assigned to a trip - note that these two notions are essentially equivalent, and that it only makes sense to assign a time slot to a trip \(i\) if the time slot has value at least equal to the length of the first segment of trip \(i\) (i.e., from the source to the central vertex).

It may help to visualize this algorithm as taking a histogram over \(\Z_{\geq0}\) of arrival times of each trip at the central vertex (under 0 delay), and drawing an arrow from each element (trip) of this histogram \emph{forward} (i.e., to a number that is no smaller) to a distinct \emph{slot} (element of \(\Z_{\geq0}\)) during which that trip will pass through the central vertex. The length of this arrow is the delay assigned to that trip.

%
\begin{algorithm}[h]
\caption{Central vertex slot allocation, with priority to longer paths.}\label{algs:star}
\begin{algorithmic}[1]
\Require\tabto{50pt}\((G:=(V,E),\lambda:E\rightarrow\mathbb{Z}_{\geq1},\mathcal{T}:=\{(s_1,t_1),\dots,(s_k,t_k)\})\)
\Ensure\tabto{50pt}\(\{(d_1,W_1),\ldots,(d_k,W_k)\}\)
\LineComment{\underline{Initialize variables}}
\State \(\mathsf{central\_vertex}\gets\mathsf{central\_vertex}(G)\)
\For{\(i\in[k]\)}
\State \(W_i\gets(s_i,\mathsf{central\_vertex},t_i)\)
\State \(a_i\gets\lambda(s_i,\mathsf{central\_vertex})\)
\State \(b_i\gets\lambda(\mathsf{central\_vertex},t_i)\)
\EndFor
\State \(\mathsf{unassigned\_trips}\gets[k]\)
\State \(\mathsf{occupied\_slots}\gets\emptyset\)
\LineComment{\underline{Schedule trips by time slots on the central vertex, with priority to longer paths}}
\While{\(\mathsf{unassigned\_trips}\neq\emptyset\)}
\State \(\mathsf{min\_useful\_slot}\gets\min_{i\in\mathsf{unassigned\_trips}}(a_i)\)
\State \(\mathsf{slot}\gets\min(\{s\in\Z_{\geq0}\mid s\geq\mathsf{min\_useful\_slot}, s\notin\mathsf{occupied\_slots}\})\)
\State \(i^\ast\gets \argmax_{\{i\in\mathsf{unassigned\_trips}\mid a_i\leq \mathsf{slot}\}}(b_i)\)
\State \(d_{i^\ast}\gets\mathsf{slot}-a_{i^\ast}\)
\State \(\mathsf{unassigned\_trips}\gets\mathsf{unassigned\_trips}\setminus\{i^\ast\}\)
\State \(\mathsf{occupied\_slots}\gets\mathsf{occupied\_slots}\cup\{\mathsf{slot}\}\)
\EndWhile
\State \textbf{return} \(\{(d_1,W_1),\ldots,(d_k,W_k)\}\)
\end{algorithmic}
\end{algorithm}
Note that for brevity, we have excluded from \cref{algs:star} the pseudocode to check that \(G\) is indeed an oriented star with central vertex \(\mathsf{central\_vertex}(G)\), and the pseudocode to handle the case where some \(s_i\) or \(t_i\) is the central vertex (for such an \(i\), \(W_i\) should instead be set to \((s_i,t_i)\), and one of \(a_i,b_i\) should be set to 0, according to whether \(s_i\) or \(t_i\) is the central vertex). Also, we allow ties encountered when setting \(i^\ast\) to be broken arbitrarily.\footnote{In our analyses that follow, we also assume all ties are broken arbitrarily.} It is clear that the algorithm runs in polynomial time, and that its output is always feasible: \(i^\ast\) is always defined (due to \(\mathsf{min\_useful\_slot}\)), all delays are nonnegative (by definition of \(i^\ast\)), and collisions - which can only occur at the central vertex - are avoided (due to \(\mathsf{occupied\_slots}\)). We now prove that \cref{algs:star} optimally solves both min-sum and min-max TDW on oriented stars.

\begin{theorem}\label{thm:ms-oriented-stars}
\cref{algs:star} is an exact poly-time algorithm for min-sum Time Disjoint Walks on oriented stars.
\end{theorem}
\begin{proof}
Let \((G,\lambda,\mathcal{T})\) be a TDW instance, where \(G\) is an oriented star with central vertex \(x\). Let \((d_1,W_1),\ldots,(d_k,W_k)\) be the solution outputted by \cref{algs:star}, and let \((d_1^\ast,W_1^\ast),\ldots,(d_k^\ast,W_k^\ast)\) be an optimal solution under the min-sum objective. For each \(i\in[k]\), define \(a_i=a_i^\ast=\lambda(W_i,x)\), \(b_i=b_i^\ast=\lambda(W_i)-\lambda(W_i,x)\), and \(y_i:= d_i+a_i, y_i^\ast := d_i^\ast + a_i^\ast\). Let \(S:=\{y_i\mid i\in[k]\}\) denote the time slots during which the central vertex is busy under our algorithm's solution, and let \(S^\ast:=\{y_i^\ast\mid i\in[k]\}\) denote the same notion under a min-sum optimal solution. By the time-disjointness constraint of TDW, we know that \(|S|=|S^\ast|=k\). Recall that the min-sum objective sums delays and walk lengths over all trips, and that a unique walk satisfies each demand in an oriented star. Using these, we know that our solution will have the same min-sum cost as the optimal min-sum cost if \(S=S^\ast\).

Assume for contradiction that \(S\neq S^\ast\). Let \(\beta\) be the smallest element in the symmetric difference of \(S\) and \(S^\ast\). If \(\beta\in S\setminus S^\ast\), then define \(T:=\{i\in[k]\mid y_i\leq\beta\}\) to contain the trips that our algorithm schedules on a slot no later than \(\beta\). Now, note that there must exist some \(i\in T\) such that \(y_i^\ast>\beta\), since \(\beta\) is the smallest element at which \(S,S^\ast\) differ, and \(i\neq j\in T\implies y_i^\ast\neq y_j^\ast\), by the time-disjointness constraint of TDW. But this means that we can reset \(d_i^\ast\) to \(\beta-a_i^\ast=\beta-a_i\) (which is guaranteed to be nonnegative, by our algorithm description) and obtain a valid solution to this TDW instance that costs less than the optimal, which is a contradiction.

If \(\beta\in S^\ast\setminus S\), then by almost identical reasoning as above, we know that if we define \(T^\ast := \{i\in[k]\mid y_i^\ast\leq \beta\}\), then there must exist some \(i\in T^\ast\) such that \(y_i>\beta\) and \(\beta-a_i\) is nonnegative. But this would mean that during the iteration of the main while loop at which \(i\) gets assigned a delay, \(\beta\geq\mathsf{min\_useful\_slot}\) and \(\beta\notin\mathsf{occupied\_slots}\) and \(\beta<y_i\), yet \(\mathsf{slot}\) is set to \(y_i\), which contradicts the definition of the min function (in our algorithm description). This completes the proof.
\end{proof}

\begin{theorem}
\cref{algs:star} is an exact poly-time algorithm for min-max Time Disjoint Walks on oriented stars.
\end{theorem}
\begin{proof}
First, define the same variables as in the proof to \cref{thm:ms-oriented-stars}, replacing the words ``min-sum'' with ``min-max'' as necessary. Our goal here will be to find a sequence of consecutive slots \(Q:=[q_\text{start},q_\text{end}]\subseteq\Z_{\geq0}\) with the following properties:
\begin{itemize}
\item \(Q\subseteq S\).
\item If we define \(I_Q:=\{i\in[k]\mid d_i+a_i\in Q\}\), then \(A_Q:=\{a_i \mid i\in I_Q\}\subseteq Q\).
\item If we define \(\gamma\in I_Q\) such that \(d_\gamma+a_\gamma=q_\text{end}\), then \(d_\gamma+a_\gamma+b_\gamma\geq d_j+a_j+b_j\), for all \(j\in[k]\).
\item For all \(j\in I_Q\), we have \(b_j\geq b_\gamma\).
\end{itemize}
Given such a \(Q\), the first two properties ensure \(|I_Q|=|Q|\) and that these \(|Q|\) trips (indexed by \(I_Q\)), under 0 delay, would arrive at the central vertex at some time in \(Q\). Thus, the min-max optimal solution (and, in general, any feasible solution) to the TDW instance must set, for some \(h\in I_Q\), \(d_h^\ast+a_h^\ast\geq q_\text{end}\) (otherwise, because delays are nonnegative, the pigeonhole principle shows that the optimal solution schedules two trips in \(I_Q\) to share a slot in \(Q\)). The third and fourth property ensure that for this same \(h\in I_Q\), \(b_h^\ast=b_h\geq b_\gamma\) and thus \(d_h^\ast+a_h^\ast+b_h^\ast \geq q_\text{end}+b_\gamma= d_\gamma+a_\gamma+b_\gamma\geq d_j+a_j+b_j\), for all \(j\in[k]\). This shows that the max trip time in any solution to the TDW instance is at least the max trip time on the solution produced by our algorithm, proving that our algorithm is optimal under the min-max constraint. We construct such a \(Q\), below.

First, define \(\gamma := \argmax_{i\in[k]}(d_i+a_i+b_i)\), and set \(q_\text{end}:=d_\gamma+a_\gamma\). Define \(\mathsf{anchor}:=\gamma, \mathsf{curr\_slot}:= q_\text{end}\), and repeat the following while \(a_\mathsf{anchor}\neq\mathsf{curr\_slot}\): decrement \(\mathsf{curr\_slot}\), let \(\alpha\in[k]\) be the index such that \(d_\alpha+a_\alpha=\mathsf{curr\_slot}\), and reset \(\mathsf{anchor}\) to \(\argmin_{i\in\{\mathsf{anchor},\alpha\}}(a_i)\). Finally, after quitting the loop, set \(q_\text{start}:=a_\mathsf{anchor}\), and set \(Q:=[q_\text{start},q_\text{end}]\).

We must argue that the above process is well-defined (in particular, that \(\alpha\) is always defined and that the process terminates) and that \(Q\) has all the properties we desire. To see the former, note that if we reach some \(\mathsf{curr\_slot}\) such that there is no \(\alpha\in[k]\) with \(d_\alpha+a_\alpha=\mathsf{curr\_slot}\), then the description of \cref{algs:star} is contradicted, because when this algorithm is picking a delay for trip \(\mathsf{anchor}\), the above analysis suggests that it picks a delay larger than \(\mathsf{curr\_slot}-a_\mathsf{anchor}\), even though \(\mathsf{curr\_slot}\geq \mathsf{min\_useful\_slot}\) and \(\mathsf{curr\_slot}\notin\mathsf{occupied\_slots}\). Thus, \(\alpha\) is always defined, and so each iteration of the loop is well defined. Because \(k\) is finite, the process must end.

Now, note that the well-definedness of \(\alpha\) and the stopping condition of the loop also shows that \(Q:=[q_\text{start},q_\text{end}]\subseteq S\). Next, \(A_Q\subseteq Q\) by the setting of \(\mathsf{anchor}\) and \(q_\text{start}\). The third desirable property of \(Q\) trivially holds by selection of \(\gamma\). All that remains is to show that for all \(j\in I_Q\), \(b_j\geq b_\gamma\). To see this, consider any \(j\in I_Q\setminus\{\gamma\}\) and the value of \(\mathsf{anchor}\) during the \(Q\)-creation process right after \(\alpha\) is set to \(j\) (i.e., before \(\mathsf{anchor}\) is potentially reset). Observe that \(d_j+a_j < d_\mathsf{anchor}+a_\mathsf{anchor}\) and that \(a_j,a_\mathsf{anchor}\leq d_j+a_j\) by the while loop condition. Thus, by description of \cref{algs:star}, we must have that \(b_\mathsf{anchor}\leq b_j\), for otherwise our algorithm would have preferred to set \(d_\mathsf{anchor}\) to \(d_j+a_j-a_\mathsf{anchor}\). Since \(\mathsf{anchor}\) is initialized to \(\gamma\), we have we have \(b_\gamma\leq b_j\), \(\forall j\in I_Q\).
\end{proof}

\subsection{Min-sum Time Disjoint Walks on runways is NP-hard}
We will reduce from min-sum coloring on interval graphs. We define this problem and state its NP-hardness,\footnote{Whenever we refer to an optimization problem as NP-hard, we mean the natural corresponding decision problem.} below.
\begin{definition}[Min-sum coloring]
Given an undirected graph \(G\), find a coloring \(c:V(G)\to\Z_{\geq1}\) such that adjacent vertices get different colors and \(\sum_{v\in V}c(v)\) is minimized.
\end{definition}
\begin{definition}[Interval graphs]
An undirected graph \(G:=(V,E)\) such that each \(v\in V\) may be identified with an interval \(S_v:=[v_\text{start},v_\text{end}]\subseteq\mathbb{R}\) such that \(xy\in E\iff S_x\cap S_y\neq\emptyset\).
\end{definition}
\begin{theorem}[\hspace{1sp}\cite{mscoloring1,mscoloring2}]Min-sum coloring of interval graphs is NP-hard, even for intervals taken over \(\Z\).
\end{theorem}
We are now ready to prove our result.
\begin{theorem}
Min-sum Time Disjoint Walks on runways is NP-hard.
\end{theorem}
\begin{proof}
Given an interval graph \(G\) and a parameter \(q\) (for which we ask whether there is a coloring of \(G\) whose sum adds up to at most \(q\)), we efficiently construct a TDW instance \((G^\prime,\lambda,\mathcal{T})\) and parameter \(q^\prime\) such that \(G\) admits a solution of (sum-)cost at most \(q\) if and only if \((G^\prime,\lambda,\mathcal{T})\) admits a solution of (sum-)cost at most \(q^\prime\).

Given an interval graph \(G\) and parameter \(q\), let \(\{S_v:=[v_\text{start},v_\text{end}]\subseteq\Z\mid v \in V\}\) denote its interval representation, let \(P=\bigcup_{v\in V}\{v_\text{start},v_\text{end}\}\) denote all the interval endpoints, and construct a bijection \(\pi:P\to[|P|]\) that orders the endpoints such that \(p_1\leq p_2\in P\implies \pi(p_1)\leq \pi(p_2)\). Clearly, we may equivalently represent \(G\) as the intervals \(\{S^\prime_v:=[\pi(v_\text{start}),\pi(v_\text{end})]\subseteq\Z_{\geq1}\mid v\in V\}\).

Now, construct a runway digraph \(G^\prime\) with arc lengths \(\lambda:E(G^\prime)\to\Z_{\geq1}\) as follows: add to \(V(G^\prime)\) the vertices \(x_1,x_2,\ldots,x_{|P|}\); turn these vertices into a directed path from \(x_1\) to \(x_{|P|}\) by adding the appropriate arcs to \(E(G^\prime)\); and let \(\lambda\) assign a length of 1 to each such arc. Then, for each \(v\in V(G)\): add to \(V(G^\prime)\) two distinct nodes \(v_L,v_R\); add to \(E(G^\prime)\) two distinct arcs \((v_L,x_{\pi(v_\text{start})})\) and \((x_{\pi(v_\text{end})},v_R)\); and set \(\lambda(v_L,x_{\pi(v_\text{start})}):=\pi(v_\text{start})\) and \(\lambda(x_{\pi(v_\text{end})},v_R):=|P|+1-\pi(v_\text{end})\). Finally, define the set of demands \(\mathcal{T}\) as \(\{(v_L,v_R)\mid v\in V(G)\}\), and let \(q^\prime:=|V(G)|\cdot|P|+q\). This completes the reduction.

To prove the correctness of our reduction, we first make several key observations. By construction, we have that for any \(i\in[|P|]\), \(W_v\) includes the vertex \(x_i\) iff \(i\in S_v^\prime\). Furthermore, for any \(i\in[|P|]\) such that \(x_i\) is hit by \(W_v\), we have \(\lambda(W_v,x_i)=i\). Thus, we see that two walks \(W_u,W_v\) share a vertex \(x_i\) iff \(S_u^\prime\cap S_v^\prime\neq\emptyset\) iff \(uv\in E(G)\), and that for any such shared vertex \(x\), \(\lambda(W_u,x)=\lambda(W_v,x)\). Finally, note that for any \(v\in V(G)\), \(\lambda(W_v)=|P|+1\).

Consider the case where \(G\) has a solution of sum-cost at most \(q\). Let \(c:V(G)\to\Z_{\geq1}\) be a proper coloring that achieves this sum. We can construct a solution \(\{(d_v,W_v)\mid v\in V(G)\}\) to TDW instance \((G^\prime,\lambda,\mathcal{T})\), where \((d_v,W_v)\) are the delay and walk for demand \((v_L,v_R)\), as follows: set \(d_v=c(v)-1\), and let \(W_v\) be the unique walk from \(v_L\) to \(v_R\) in \(G^\prime\). By our key observations, we have \(\sum_{v\in V(G)}(d_v+\lambda(W_v))=\sum_{v\in V(G)}(c(v)-1)+|V(G)|\cdot(|P|+1)\leq q+|V(G)|\cdot|P|=q^\prime\), and we know that if two walks \(W_u,W_v\) share a vertex \(x\), then \(c(u)\neq c(v)\) (by the validity of \(c\)), and so \(d_u\neq d_v\). Thus, \(d_u+\lambda(W_u,x)\neq d_v+\lambda(W_v,x)\), and so this solution to the TDW instance is valid.

Consider the case where \((G^\prime,\lambda,\mathcal{T})\) has a solution of sum-cost at most \(q^\prime\). Let \(\{(d_v,W_v)\mid v\in V(G)\}\) be such a solution. We construct a solution \(c:V(G)\to\Z_{\geq1}\) to min-sum coloring instance \(G\) by setting \(c(v):=d_v+1\). By our key observations, we know that \(uv\in V(G)\implies W_u,W_v\) share some vertex \(x\) where \(\lambda(W_u,x)=\lambda(W_v,x)\); thus, by the validity of our TDW solution, we know that \(d_u\neq d_v\), which means \(c(u)\neq c(v)\). Thus, \(c\) is a valid coloring of \(V(G)\), and \(\sum_{v\in V(G)}c(v)=\sum_{v\in V(G)}(d_v+1)\leq q^\prime - (|P|+1)\cdot|V(G)| +|V(G)|=q\). Thus, our reduction is correct, and it can clearly be done in poly time.
\end{proof}

\section{Approximation preliminaries}\label{section:approx-prelims}

Given an optimization problem \(\mathcal{P}\), we let \(I_\mathcal{P}\) denote the instances of \(\mathcal{P}\), \(SOL_\mathcal{P}\) map each \(x\in I_{\mathcal{P}}\) to a set of feasible solutions, and let \(c_\mathcal{P}\) assign a real cost to each pair \((x,y)\) where \(x\in I_\mathcal{P}\) and \(y\in SOL_\mathcal{P}(x)\). For \(x\in I_\mathcal{P}\), we let \(OPT_\mathcal{P}(x) := \min_{y^\ast\in SOL_{\mathcal{P}}(x)}c_\mathcal{P}(x,y^\ast)\) if \(\mathcal{P}\) is a minimization problem, and \(OPT_\mathcal{P}(x) := \max_{y^\ast\in SOL_{\mathcal{P}}(x)}c_\mathcal{P}(x,y^\ast)\) otherwise.

If \(\mathcal{A}\) is a polynomial time algorithm with input \(x\in I_\mathcal{P}\) and output \(y\in SOL_\mathcal{P}(x)\), we say that \(\mathcal{A}\) is a \emph{\(\rho\)-approximation algorithm}, or \emph{has approximation ratio} \(\rho\), if \(\mathcal{P}\) is a minimization problem and \(c_\mathcal{P}(x,\mathcal{A}(x))/OPT_\mathcal{P}(x)\leq \rho\), or \(\mathcal{P}\) is a maximization problem and \(OPT_\mathcal{P}(x)/c_\mathcal{P}(x,\mathcal{A}(x))\leq \rho\), for all \(x\in I_\mathcal{P}\). Note that \(\rho\geq1\).

The class \emph{APX} contains all optimization problems that admit a \(\rho\)-approximation algorithm, for \emph{some} constant \(\rho>1\). An optimization problem is said to be \emph{APX-hard} if every problem in APX can be reduced to it through an approximation-preserving reduction. One reduction of this type is the \emph{L-reduction}:
\begin{definition}[L-reduction]
An L-reduction from an optimization problem \(\mathcal{P}\) to an optimization problem \(\mathcal{Q}\), denoted \(\mathcal{P}\leq_L\mathcal{Q}\), is a tuple \((f,g,\alpha,\beta)\), where:
\begin{itemize}
\item For each \(x\in I_\mathcal{P}\), \(f(x)\in I_\mathcal{Q}\) and can be computed in polynomial time.
\item For each \(y\in SOL_\mathcal{Q}(f(x))\), \(g(x,y)\in SOL_\mathcal{P}(x)\) and can be computed in polynomial time.
\item \(\alpha\) is a positive real constant such that for each \(x\in I_\mathcal{P}\),
\[
OPT_\mathcal{Q}(f(x))\leq\alpha\cdot OPT_\mathcal{P}(x).
\]
\item \(\beta\) is a positive real constant such that for each \(x\in I_\mathcal{P},y\in SOL_\mathcal{Q}(f(x))\),
\[
\big|OPT_\mathcal{P}(x)-c_\mathcal{P}(x,g(x,y))\big|\leq \beta\cdot\big|OPT_\mathcal{Q}(f(x))-c_\mathcal{Q}(f(x),y)\big|.
\]
\end{itemize}
\end{definition}
If a problem is APX-hard, it is NP-hard to \(\rho\)-approximate for some constant \(\rho>1\); thus, showing APX-hardness is strictly stronger than showing NP-hardness. To show APX-hardness, one can simply L-reduce from a known APX-hard problem. We refer the reader to~\cite{approx-textbook} for a good reference on approximation.

\section{Hardness of approximation}\label{section:hardness}
To show the hardness of min-sum TDW on bounded-degree DAGs, we show an L-reduction from MAX-E2SAT(3), which is known to be APX-hard~\cite{maxsat}. We remind the reader of the definition, below, and then proceed with our proof.

\begin{definition}[MAX-E2SAT(3)] Let \(\phi\) be a CNF formula in which (i) each clause contains exactly two literals on distinct variables, and (ii) each variable appears in at most three clauses. Find a truth assignment to the variables in \(\phi\) that maximizes the number of satisfied clauses.
\end{definition}
\begin{theorem}
Min-sum Time Disjoint Walks is APX-hard, even for DAGs with \(\Delta\leq3\).
\end{theorem}
\begin{proof}
We let \(\mathcal{P}:=\) MAX-E2SAT(3), \(\mathcal{Q}:=\) TDW with instances restricted to those containing DAGs with \(\Delta\leq3\), and show that \(\mathcal{P}\leq_L\mathcal{Q}\). Below, we describe our L-reduction \((f,g,\alpha,\beta)\).

\textbf{Description of \(f\)}: Given an instance \(\phi\in I_\mathcal{P}\) with \(n\) variables and \(m\) clauses, we let \(X:=\{x_1,\dots,x_n\}\) refer to its variables and \(\mathcal{C}:=\{C_1,\dots,C_m\}\) refer to its clauses. We let \(L:=\{x_1,\dots,x_n,\overline{x}_1,\dots,\overline{x}_n\}\) refer to its literals. For convenience, we define \(e:L\to X\) that extracts the variable from a given literal; i.e., \(e(x_i)=e(\overline{x}_i)=x_i\). We label the literals in clause \(C_j\) as \(l^1_j,l^2_j\). For each \(l\in L\), we let \(S_l := \{l^a_j\mid a\in[2],j\in[m],l^a_j=l\}\) capture all occurrences of literal \(l\) in \(\phi\). Finally, for each \(l\in L\), we define an arbitrary bijection \(\pi_l:S_l\to[|S_l|]\) to induce an ordering on \(S_l\). We will let \(\pi_l^{-1}\) denote its inverse: i.e., \(\pi_l^{-1}(1)\) is the first element in \(S_l\) in the order induced by \(\pi_l\).

We may now describe \(f\), which constructs an instance \((G,\lambda,\mathcal{T})\in I_\mathcal{Q}\) from \(\phi\). We start with the construction of \(G\) (see \cref{fig:gadgets}), which closely follows the standard proof of NP-hardness for Disjoint Paths: for each clause \(C_j=(l_j^1\lor l_j^2)\) in \(\phi\), we create a new \emph{clause gadget} and add it to \(G\). That is, for each clause \(C_j\), we add the following vertex and arc set to our construction:
\begin{align*}
V_{C_j} &:= \{c_j^s,l_j^1,l_j^2,l_j^{1^\prime},l_j^{2^\prime},c_j^t\}\\
E_{C_j} &:= \{(c_j^s,l_j^1),(c_j^s,l_j^2),(l_j^1,l_j^{1^\prime}),(l_j^2,l_j^{2^\prime}),(l_j^{1^\prime},c_j^t),(l_j^{2^\prime},c_j^t)\}
\end{align*}

Next, for each \(x_i\in X\), we add an \emph{interleaving variable gadget} as follows: first, we add two vertices \(x_i^s,x_i^t\) to \(V(G)\). Then, we wish to create exactly two directed paths (walks), \(W_{x_i}^+,W_{x_i}^-\), from \(x_i^s\) to \(x_i^t\): we want \(W_{x_i}^+\) to travel through all vertices corresponding to positive literals of \(x_i\), and \(W_{x_i}^-\) to travel through all vertices corresponding to negative literals of \(x_i\). Formally, for each \(l\in\{x_i,\overline{x}_i\}\), we create a path from \(x_i^s\) to \(x_i^t\) as follows. First, if \(|S_l|=0\), we add arc \((x_i^s,x_i^t)\) to \(E(G)\). Otherwise, we add arcs \((x_i^s,\pi_l^{-1}(1))\) and \(((\pi_l^{-1}(|S_l|))^\prime,x_i^t)\) to \(E(G)\), and then for each \(j\in[|S_l|-1]\), we add arc \(((\pi_l^{-1}(j))^\prime,\pi_l^{-1}(j+1))\). Note that the prime symbols are merely labels, and are used in our construction to ensure that the max degree of \(G\) remains at most three.
\begin{figure}[t]
\centering
\begin{tikzpicture}[line width=1pt]
\tikzset{vertex/.style = {shape=circle,draw,minimum size=0.5em, inner sep=0pt}}
\tikzset{edge/.style = {->,> = latex'}}

\node[vertex,label=left:\(x_i^s\)] (xis) at (1,0) {};

\node[vertex,label=above:\(c_{j_1}^s\)] (c1s) at (3.5,1.5) {};
\node[vertex,label={left,yshift=1em:\(l_{j_1}^{1}\)}] (c1l1) at (3,0.5) {};
\node[vertex,label={left,yshift=-1em:\(l_{j_1}^{1^\prime}\)}] (c1l1p) at (3,-0.5) {};
\node[vertex,label={right,yshift=1em:\(l_{j_1}^{2}\)}] (c1l2) at (4,0.5) {};
\node[vertex,label={right,yshift=-1em:\(l_{j_1}^{2^\prime}\)}] (c1l2p) at (4,-0.5) {};
\node[vertex,label=below:\(c_{j_1}^t\)] (c1t) at (3.5,-1.5) {};

\node[vertex,label=above:\(c_{j_2}^s\)] (c2s) at (6.5,1.5) {};
\node[vertex,label={left,yshift=1em:\(l_{j_2}^{1}\)}] (c2l1) at (6,0.5) {};
\node[vertex,label={left,yshift=-1em:\(l_{j_2}^{1^\prime}\)}] (c2l1p) at (6,-0.5) {};
\node[vertex,label={right,yshift=1em:\(l_{j_2}^{2}\)}] (c2l2) at (7,0.5) {};
\node[vertex,label={right,yshift=-1em:\(l_{j_2}^{2^\prime}\)}] (c2l2p) at (7,-0.5) {};
\node[vertex,label=below:\(c_{j_2}^t\)] (c2t) at (6.5,-1.5) {};

\node[vertex,label=above:\(c_{j_3}^s\)] (c3s) at (9.5,1.5) {};
\node[vertex,label={left,yshift=1em:\(l_{j_3}^{1}\)}] (c3l1) at (9,0.5) {};
\node[vertex,label={left,yshift=-1em:\(l_{j_3}^{1^\prime}\)}] (c3l1p) at (9,-0.5) {};
\node[vertex,label={right,yshift=1em:\(l_{j_3}^{2}\)}] (c3l2) at (10,0.5) {};
\node[vertex,label={right,yshift=-1em:\(l_{j_3}^{2^\prime}\)}] (c3l2p) at (10,-0.5) {};
\node[vertex,label=below:\(c_{j_3}^t\)] (c3t) at (9.5,-1.5) {};

\node[vertex,label=right:\(x_i^t\)] (xit) at (12,0) {};

\draw[edge] (c1s) to (c1l1);
\draw[edge] (c1s) to (c1l2);
\draw[edge] (c1l1) to (c1l1p);
\draw[edge,red] (c1l2) to (c1l2p);
\draw[edge] (c1l1p) to (c1t);
\draw[edge] (c1l2p) to (c1t);

\draw[edge] (c2s) to (c2l1);
\draw[edge] (c2s) to (c2l2);
\draw[edge] (c2l1) to (c2l1p);
\draw[edge,green] (c2l2) to (c2l2p);
\draw[edge] (c2l1p) to (c2t);
\draw[edge] (c2l2p) to (c2t);

\draw[edge] (c3s) to (c3l1);
\draw[edge] (c3s) to (c3l2);
\draw[edge,green] (c3l1) to (c3l1p);
\draw[edge] (c3l2) to (c3l2p);
\draw[edge] (c3l1p) to (c3t);
\draw[edge] (c3l2p) to (c3t);
\draw[edge,red,bend right=9] (xis) to (c1l2);
\draw[edge,red,bend left=12] (c1l2p) to (xit);
\draw[edge,green,bend right=9] (xis) to (c2l2);
\draw[edge,green,bend left=15] (c2l2p) to (c3l1);
\draw[edge,green,bend left=15] (c3l1p) to (xit);

\end{tikzpicture}
\caption{An interleaving variable gadget (and its affiliated clause gadgets) corresponding to a variable with one negative occurrence (red path) and two positive occurrences (green path).}
\label{fig:gadgets}
\end{figure}
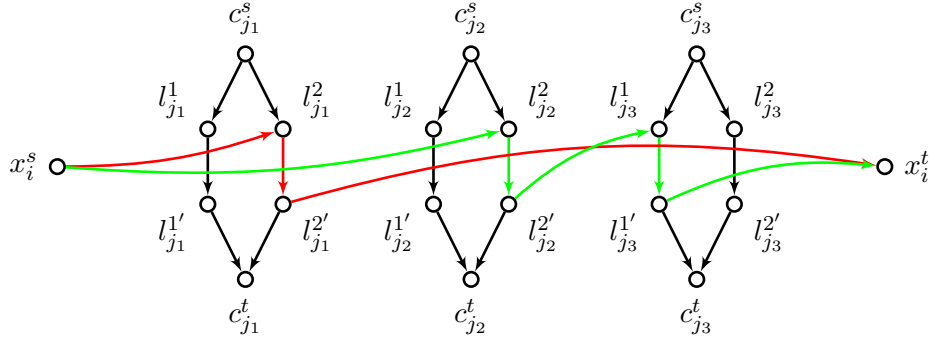
This completes our construction of \(G\). We now define a set of \(n+m\) demands, where each corresponds to a variable or a clause:
\[
\mathcal{T} := \{(x_i^s,x_i^t)\mid i\in[n]\}\cup\{(c_j^s,c_j^t)\mid j\in[m]\}
\]

Finally, we must define arc lengths \(\lambda:E\rightarrow\mathbb{Z}_{\geq1}\). We will do this in a way that for each \(j\in[m],a\in[2]\), we have \(\lambda(W_{c_j^s},l_j^a)=\lambda(W_{x^s_i},l_j^a)\), where \(W_{c_j^s}\) is the unique walk in \(G\) from \(c_j^s\) to \(l_j^a\), and \(W_{x_i^s}\) is the unique walk in \(G\) from \(x_i^s=e(l_j^a)^s\) to \(l_j^a\). Call this property (\(\ast\)). To facilitate our analysis, we will also want every demand-satisfying path in \(G\) to have the same length.

Since \(\phi\) is an instance of MAX-E2SAT(3), we know that for each \(i\in[n]\), each of the two paths between \(x_i^s\) and \(x_i^t\) passes through at most 3 clause gadgets. Thus, by our construction, each such path includes at most 7 arcs, and any path from a variable \(x_i^s\) to some literal \(l_j^a\) with \(x_i=e(l_j^a)\) can use at most 5 arcs. Thus, we can successfully force each demand-satisfying path in \(G\) to have length 7 while maintaining property (\(\ast\)) by defining \(\lambda:E\rightarrow\mathbb{Z}_{\geq1}\) as follows, completing our construction of \((G,\lambda,\mathcal{T})\in I_\mathcal{Q}\):
\begin{align*}
\lambda(u,v):=\begin{cases}
1,&\textbf{if }(u,v)=(l_j^a,l_j^{a^\prime}), j\in[m],a\in[2];\textbf{ or}\\
&\textbf{if }(u,v)=(x_i^s,l_j^a),i\in[n],j\in[m],a\in[2];\textbf{ or}\\
&\textbf{if }(u,v)=(l_h^{a^\prime},l_j^{b}),h,j\in[m], a,b\in[2];\\
7,&\textbf{if }(u,v)=(x_i^s,x_i^t),i\in[n]\\
7-2|S_l|,&\textbf{if }(u,v)=(l_j^{a^\prime},x_i^t),j\in[m],a\in[2],i\in[n],l_j^{a}=l\\
2\pi_l(l_j^a)-1,&\textbf{if }(u,v)=(c_j^s,l_j^a),j\in[m],a\in[2],l_j^a=l\\
7-1-\lambda(c_j^s,l_j^a),&\textbf{if }(u,v)=(l_j^{a^\prime},c_j^t),j\in[m],a\in[2]
\end{cases}
\end{align*}

\textbf{Description of \(g\)}: Given a solution \(y\in SOL_\mathcal{Q}(f(\phi))\), we construct a solution \(g(\phi,y)\in SOL_\mathcal{P}(\phi)\) through two consecutive transformations: \(z\), followed by \(q\). That is, we will define transformations \(z\) and \(q\) such that \(g\) is the composition \(g(\phi,y) := q(\phi,z(y))\).

We define \(z\) to transform solution \(y\) into another solution \(y^\prime\in SOL_\mathcal{Q}(f(\phi))\) such that \(c_\mathcal{Q}(f(\phi),y^\prime)\leq c_\mathcal{Q}(f(\phi),y)\) and such that \(y^\prime\) assigns 0 delay to demands associated with interleaving variable gadgets. To accomplish this, recall that \(y=\{(d_1,W_1),\dots,(d_{n+m},W_{n+m})\}\), by definition of \(SOL_\mathcal{Q}\). Without loss of generality, we may assume tuples indexed with \([n]\) correspond to demands on interleaving variable gadgets, and tuples indexed with \([n+m]\setminus[n]\) correspond to demands on clause gadgets.

Now, while there exists some \(i\in[n]\) such that \(d_i>0\) (and thus \(d_i\geq1\)), we perform the following modification on \(y\): first, we reset \(W_i\) to be the path traveling through at most one clause gadget - the positive or negative path \emph{must} have this property, because each variable appears in \(\phi\) at most three times, by definition of MAX-E2SAT(3). Now, reset \(d_i\) to 0. If \(W_i\) shares a vertex with another walk \(W_j\), we know \(j\in[n+m]\setminus[n]\), by construction of \(G\). In this case, reset \(d_j\) to 1 if and only if \(d_j\) is currently 0. By construction of \(\lambda\), the walks remain time disjoint and the cost of the solution does not increase.

In the second transformation, \(q\), we transform modified solution \(y^\prime\) into an assignment \((A:X\rightarrow\{T,F\})\in SOL_\mathcal{P}(\phi)\) as follows: for each \(i\in[n]\), set \(A(x_i)=T\) if and only if \(W_{x_i}\), the walk from \(x_i^s\) to \(x_i^t\), takes the negative literal path.

\textbf{Valid value for \(\alpha\)}: We will show that for \(\alpha=29\), \(OPT_\mathcal{Q}(f(\phi))\leq \alpha\cdot OPT_\mathcal{P}(\phi)\). To see this, we make two observations. First observation: if \(A:X\rightarrow\{T,F\}\) is a truth assignment for \(\phi\), then we can construct a solution to \(f(\phi)\) as follows: for each \(i\in[n]\), connect demand \((x_i^s,x_i^t)\) using the negative literal path if \(A(x_i)=T\), and the positive literal path if \(A(x_i)=F\). Either way, assign a delay of 0. Then, for each \(j\in[m]\) where clause \(C_j\) is satisfied by assignment \(A\), connect demand \((c_j^s,c_j^t)\) using a walk that goes through a literal that evaluates to true under \(A\). Assign a delay of 0 to this demand. For each clause \(C_j\) that isn't satisfied by \(A\), select an arbitrary walk to complete the corresponding demand \((c_j^s,c_j^t)\). Assign a delay of 1 to this demand. It is clear that this is a valid solution to \(f(\phi)\). Furthermore, the cost of our solution is \(7(n+m)+U(A,\phi)\), where \(U(A,\phi)\) is the number of clauses in \(\phi\) unsatisfied by \(A\). Second observation: by linearity of expectation, if \(\phi\) is an instance of MAX-E2SAT(3), then there must exist an assignment \(A:X(\phi)\rightarrow\{T,F\}\) that satisfies at least \(3/4\) of the clauses. 

We may now prove the desired inequality for \(\alpha=29\). From our first observation and the fact that \(n\leq 2m\) (since each of the \(m\) clauses has 2 literals),
\begin{align}\label[ineq]{eq:opt-ineq-lemma}
OPT_\mathcal{Q}(f(\phi)) \leq 7(n+m) + (m-OPT_\mathcal{P}(\phi))\leq 22m-OPT_\mathcal{P}(\phi).
\end{align}
Now, by our second observation, we know \(OPT_\mathcal{P}(\phi)\geq 3m/4\). Thus, we have:
\begin{align*}
OPT_\mathcal{Q}(f(\phi)) \leq 22\cdot(4/3)\cdot OPT_\mathcal{P}(\phi)-OPT_\mathcal{P}(\phi)\leq 29\cdot OPT_\mathcal{P}(\phi).
\end{align*}

\textbf{Valid value for \(\beta\)}: We will show that for \(\beta=1\) and any \(y\in SOL_\mathcal{Q}(f(\phi))\), \(\big(OPT_\mathcal{P}(\phi)-c_\mathcal{P}(\phi,g(\phi,y))\big)\leq\beta\cdot\big(c_\mathcal{Q}(f(\phi),y)-OPT_\mathcal{Q}(f(\phi))\big)\), as required. As a first step, we recall that transformations \(z,q\) define \(g\), and let \(\gamma\) denote the number of clause gadget demands assigned a delay of 0 by solution \(z(y)\) to \(f(\phi)\). We make the following crucial claim:
\begin{align}\label[ineq]{eq:crucial-claim}
c_\mathcal{P}(\phi,g(\phi,y)):=c_\mathcal{P}(\phi,q(\phi,z(y))) \geq \gamma.
\end{align}
To see this, note the following: by construction, \(z(y)\) is a valid solution to \(f(\phi)\). Thus, if \(z(y)\) assigns clause gadget demand \((c_j^s,c_j^t)\) a delay \(d_j=0\) and walk \(W_j\) that passes through literal \(l\), then \(l\) is a positive literal if and only if the walk selected for the interleaving variable gadget demand \((x_i^s,x_i^t)\) (where \(x_i=e(l)\)) does not travel through the positive literals of \(x_i\). By definition of \(q\), this occurs if and only if \(g(\phi,y)\) assigns \emph{true} to \(x_i\). Thus, a clause gadget demand given 0 delay by \(z(y)\) corresponds to a clause in \(\phi\) satisfied by \(g(\phi,y)\), thus proving \cref{eq:crucial-claim}.

Next, by definition of \(\gamma\) and \(z\), we have:
\begin{align}\label[ineq]{eq:gamma-z}
7(n+m)+(m-\gamma)\leq c_\mathcal{Q}(f(\phi),z(y))\leq c_\mathcal{Q}(f(\phi),y).
\end{align}
Combining \cref{eq:crucial-claim,eq:gamma-z}, we get:
\begin{align}\label[ineq]{eq:super-ineq}
c_\mathcal{P}(\phi,g(\phi,y))\geq \gamma\geq 7n+8m-c_\mathcal{Q}(f(\phi),y).
\end{align}
Finally, \cref{eq:opt-ineq-lemma,eq:super-ineq} give us:
\begin{align*}
OPT_\mathcal{P}(\phi)-c_\mathcal{P}(\phi,g(\phi,y)) &\leq \big(7n+8m-OPT_\mathcal{Q}(f(\phi))\big)-\big(7n+8m-c_\mathcal{Q}(f(\phi),y)\big)
\\
&=\beta\cdot\big(c_\mathcal{Q}(f(\phi),y)-OPT_\mathcal{Q}(f(\phi))\big),
\end{align*}
for \(\beta=1\), as desired. This completes the proof.
\end{proof}
\begin{remark}
It is straightforward to tweak the above construction and analysis to show that min-max Time Disjoint Walks is NP-hard on DAGs with \(\Delta\leq 3\).
\end{remark}

\section{Approximation algorithm}\label{section:approx-alg}
\subsection{Algorithm}
We present \cref{algs:clever}, which approximates TDW by finding shortest paths to satisfy each demand, and then greedily assigning delays to each trip (with priority given to shorter trips). To simplify notation, \emph{we assume that the inputted terminal pairs are ordered by nondecreasing shortest path length} (if not, we may simply sort the indices after finding the shortest demand-satisfying paths). The algorithm clearly runs in polynomial time, and the \(\mathsf{bad\_delay}\) variables ensure the feasibility of its output. Next, we briefly note the following easy bound:

\begin{algorithm}[h]
\caption{Shortest paths \& greedy delays, with priority to shorter paths.}\label{algs:clever}
\begin{algorithmic}[1]
\Require\tabto{50pt}\(x := (G:=(V,E),\lambda:E\rightarrow\mathbb{Z}_{\geq1},\mathcal{T}:=\{(s_1,t_1),\dots,(s_k,t_k)\})\in I_{TDW}\)
\Ensure\tabto{50pt}\(y\in SOL_{TDW}(x)\)
\State \(y\gets\{\}\)
\LineComment{\underline{Get shortest paths and dummy delays}:}
\For{\(i\in[k]\)}
    \State \(W_i\gets \mathsf{Dijkstra}(G,\lambda,s_i,t_i)\)
    \State \(d_i\gets0\)
    \State \(y\gets y\cup(d_i,W_i)\)
\EndFor
\LineComment{\underline{Greedily assign delays, with priority given to shorter paths}:}
\For{\(i\in[k]\)}
\State \(\mathsf{bad\_delays}_{i}\gets\{\}\)
\For{\(h\in[i-1]\)}
\State \(\mathsf{bad\_delays}_{i,h}\gets\{\}\)
\For{\(v\in W_{h}\cap W_{i}\)}
\State \(\mathsf{bad\_delay}\gets (d_{h}+\lambda(W_{h},v)-\lambda(W_{i},v))\)
\State \(\mathsf{bad\_delays}_{i,h}\gets \mathsf{bad\_delays}_{i,h}\cup\{\mathsf{bad\_delay}\}\)
\EndFor
\State \(\mathsf{bad\_delays}_{i}\gets \mathsf{bad\_delays}_{i}\cup \mathsf{bad\_delays}_{i,h}\)
\EndFor
\State \(d_{i}\gets\min(\mathbb{Z}_{\geq0}\setminus\mathsf{bad\_delays}_i)\)
\EndFor
\State \textbf{return} \(y\)
\end{algorithmic}
\end{algorithm}

\begin{proposition}\label[prop]{prop:very-easy-bound}
\cref{algs:clever} has an approximation ratio of \(O(k)\) for min-sum Time Disjoint Walks on general digraphs.
\end{proposition}
\begin{proof}
Let \(x:=(G,\lambda,\mathcal{T})\in I_{TDW}\), and let \(\mathcal{A}(x)\in SOL_{TDW}\) be the output of \cref{algs:clever} on \(x\). First, we show by induction that for each \(i\in[k]\),
\begin{align*}
d_{i}\leq2\sum_{h\in[i-1]}\lambda(W_{h}).
\end{align*}
For the base case \(i=1\), note that \(\mathsf{bad\_delays}_{i}=\emptyset\) and so \(d_{i}\) = 0. For \(i>1\), first observe that by definition of \(\mathsf{bad\_delay}\), we have \(d_{i}\leq1+\max_{h\in[i-1]}(d_{h}+\lambda(W_{h}))\). Thus,
\begin{align*}
d_{i}\leq&\ 1+\max_{h\in[i-1]}\bigg(2\sum_{h^\prime\in[h-1]}\lambda(W_{h^\prime})+\lambda(W_{h})\bigg) \tag*{\text{(induction hypothesis)}}\\
\leq&\ 1+2\sum_{h^\prime\in[i-2]}\lambda(W_{h^\prime}) + \lambda(W_{i-1}) \tag*{\text{(pick \(h=i-1\))}}\\
\leq&\ 2\sum_{h^\prime\in[i-1]}\lambda(W_{h^\prime}), \tag*{\text{(trips have length \(\geq 1\))}}
\end{align*}
completing the induction. Now, recallling that our algorithm uses the shortest paths to satisfy each demand, and that it assigns delays to shorter paths first, we can bound the approximation ratio as follows:
\begin{align*}
\rho \leq \frac{c_{TDW}(x,\mathcal{A}(x))}{OPT_{TDW}(x)} &\leq \frac{\sum_{i\in[k]}(d_{i}+\lambda(W_{i}))}{\sum_{i\in[k]}\lambda(W_{i})} \leq 1 + \frac{2\sum_{i\in[k]}\sum_{h\in[i-1]}\lambda(W_{h})}{\sum_{i\in[k]}\lambda(W_{i})}\\
&\leq 1 + \frac{2k\sum_{i\in[k]}\lambda(W_{i})}{\sum_{i\in[k]}\lambda(W_{i})} = O(k).
\end{align*}
\end{proof}
\begin{remark}
It is straightforward to tweak the above analysis to show that for min-max Time Disjoint Walks, \cref{algs:clever} has an approximation ratio of \(O(k)\) on general digraphs.
\end{remark}

\subsection{Analysis on bounded-degree DAGs}
We now show that our algorithm is able to achieve a better approximation ratio on bounded-degree DAGs. In what follows, we call a directed graph a ``\((2,l)\)-in-tree'' if it is a perfect binary tree of depth \(l\), in which every arc points toward the root. Analogously, a ``\((2,l)\)-out-tree'' is a perfect binary tree of depth \(l\), in which every arc points away from the root.

\begin{theorem}\label{thm:best-ratio}
\cref{algs:clever} achieves an approximation ratio of \(\Theta(k/\log k)\) for min-sum Time Disjoint Walks on bounded-degree DAGs.
\end{theorem}
\begin{proof}
\underline{Upper bound}:
Let \(x:=(G,\lambda,\mathcal{T})\in I_{TDW}\) such that \(G\) is a bounded-degree DAG. Let \(\mathcal{A}(x)\in SOL_{TDW}\) be the output of \cref{algs:clever} on \(x\). We may assume \(\Delta>1\), because otherwise it is straightforward to show \(\mathcal{A}(x)\) is optimal. In what follows, we will justify the following string of inequalities that proves the upper bound:
\begin{align}
\rho \leq \frac{c_{TDW}(x,\mathcal{A}(x))}{OPT_{TDW}(x)} &\leq \frac{\sum_{i\in[k]}\big(d_i+\lambda(W_i)\big)}{\sum_{i\in[k]}\lambda(W_i)}\label[ineq]{eq:FIRST}\\
&\leq1+\frac{d_{i^\ast}}{\lambda(W_{i^\ast})}, i^\ast:=\argmax_{i\in[k]}\bigg(\frac{d_i}{\lambda(W_i)}\bigg)\label[ineq]{eq:SECOND}\\
&\leq1+O(1)\cdot\frac{d_{i^\ast}}{\log d_{i^\ast}}\label[ineq]{eq:THIRD}\\
&\leq 1+O(1)\cdot\frac{k}{\log k}=O(k/\log k).\label[ineq]{eq:FOURTH}
\end{align}
\Cref{eq:FIRST} is clear, because our algorithm takes the shortest path to satisfy each demand. \Cref{eq:SECOND} follows (by induction) from the following general observation: given \(d_1,d_2\in\mathbb{Z}_{\geq0}\) and \(\lambda_1,\lambda_2\in\mathbb{Z}_{\geq1}\), observe \(d_1/\lambda_1\leq d_2/\lambda_2\implies (d_1+d_2)/(\lambda_1+\lambda_2)\leq d_2/\lambda_2\), and thus \((d_1+d_2)/(\lambda_1+\lambda_2)\leq\max(d_1/\lambda_1,d_2/\lambda_2)\).

To show \cref{eq:THIRD}, we need two observations. We first observe that for each \(i\in[k]\):
\begin{align*}
d_{i}\leq|\mathsf{bad\_delays}_{i}|\leq|\{h\in[i-1]\mid W_{h}\cap W_{i}\neq\emptyset\}|=:\mu_i
\end{align*}
To see this, suppose for contradiction that there exists some \(h\in[i-1]\) with \(W_{h}\cap W_{i}\neq\emptyset\) and \(|\mathsf{bad\_delays}_{i,h}| > 1\). Then, by definition of \(\mathsf{bad\_delay}\), there exist vertices \(u,v\in W_h\cap W_{i}\) and delays \(\delta_u\neq\delta_v\in\mathbb{Z}_{\geq0}\) such that:
\begin{align*}
\delta_u+\lambda(W_{i},u)&=d_{h}+\lambda(W_{h},u),\\
\delta_v+\lambda(W_{i},v)&=d_{h}+\lambda(W_{h},v),\\
\lambda(W_{h},u)
-\lambda(W_{h},v)
&=\lambda(W_{i},u)-
\lambda(W_{i},v)
+(\delta_u-\delta_v),
\end{align*}
where the last equality follows from the first two. But because \(\delta_u\neq\delta_v\), this implies that the length of the path that \(W_{h}\) and \(W_{i}\) use to travel between \(u\) and \(v\) is not the same. Because \(G\) is a DAG, \(W_{h}\) and \(W_{i}\) must visit \(u\) and \(v\) in the same order, implying that one of these walks is not taking the shortest path from \(u\) to \(v\), which contradicts the definition of the algorithm. Because \(|\mathsf{bad\_delays}_{i,h}|=0\) if \(W_{h}\cap W_{i}=\emptyset\), we have \(d_{i}\leq\mu_i\).

Next, we observe that:
\begin{align*}
\mu_i\leq\min(\Delta^{4\lambda(W_{i})},k).
\end{align*}
Showing \(\mu_i\leq k\) is trivial, by definition of \(\mu_i\) and because \(i\in[k]\). To show \(\mu_i\leq\Delta^{4\lambda(W_{i})}\), first note that in a digraph with max degree \(\Delta\), the number of paths that (i) have \(z\) arcs, (ii) start at distinct vertices, and (iii) all end at a common vertex, is upper bounded by \(\Delta^{z}\) (this is easy to show by induction on \(z\)). Thus, the number of paths with \(\leq z\) arcs, in addition to properties (ii) and (iii), is upper bounded by \(\sum_{l=0}^z\Delta^l\leq\Delta^{z+1}\), for \(\Delta>1\). Call this lemma (\(\ast\)).

Now, note that for each \(h\in[i-1]\) we may consider each \(W_{h}\) to terminate once it first hits a vertex in \(W_{i}\) (i.e., cut off all vertices that are hit afterwards) without changing the value of \(\mu_i\). Now, recall the following facts about our problem and algorithm: (I) each inputted demand has a unique source; (II) each edge in our digraph has length \(\geq1\); (III) for all \(h\in[i-1]\), \(\lambda(W_{h})\leq\lambda(W_{i})\). Thus, by (III) and lemma (\(\ast\)), each vertex in \(W_{i}\) can be hit by at most \(\Delta^{\lambda(W_{i})+1}\) walks in \(\{W_{1},\dots,W_{i-1}\}\). Furthermore, (II) tells us that the number of vertices in \(W_{i}\) is no more than \(\lambda(W_i)+1\). Thus, recalling that \(\Delta>1\), and that no demands have the same source and destination, and (II), we see that
\[
\mu_i \leq (\lambda(W_{i})+1)\Delta^{\lambda(W_{i})+1}\leq\Delta^{2(\lambda(W_{i})+1)}\leq\Delta^{4\lambda(W_{i})},
\]
as desired. We now note that we may assume \(d_{i}\) is greater than any constant (otherwise, \cref{eq:SECOND} automatically proves a constant approximation ratio, completing the proof). Thus, from this and the above observations, we have \(\log(d_{i})\leq 4\lambda(W_{i})\log(\Delta)\). This proves \cref{eq:THIRD}, because our graph has bounded degree.

\Cref{eq:FOURTH} is not difficult: as stated above, we will always have \(d_{i}\leq k\), and we may always assume \(d_{i}\geq 3\). Basic calculus shows the function \(x/\log x\) increases over \(x\geq3\).

\underline{Lower bound}: We show \(\forall l\in\N_{\geq1},k:=2^l\), \(\exists (G_k,\lambda_k,\mathcal{T}_k)\in I_{TDW}\) such that \(G_k\) is a bounded-degree DAG and \cref{algs:clever} achieves an approximation ratio of \(\Omega(k/\log k)\). Construct \(G_k\) by taking a \((2,l)\)-in-tree \(A_S\) and a \((2,l)\)-out-tree \(A_T\). Draw an arc from the root of the former to the root of the latter. Then, arbitrarily pair each leaf (source) in \(A_S\) with a unique leaf (destination) in \(A_T\). For each such pair, draw an arc from source to destination (called a ``bypass arc''), and add a demand to \(\mathcal{T}_k\). Finally, define \(\lambda_k\) to assign length \(1+2l\) to each ``bypass'' arc, and length 1 to all other arcs. We refer the reader to \cref{fig:master-figure}, part (i).

We may assume our algorithm does not satisfy demands using the bypass arcs (as all demand-satisfying paths have length \(2l+1\), and no tie-breaking scheme is specified). Thus, each demand-satisfying path uses the root of \(A_S\), which incurs a total delay of \(0+1+\ldots+(k-1)=\Omega(k^2)\) and total path length of \(k\cdot(1+2l)\). Had the bypass arcs been used, no delay would have been required, and the total path length would have still been \(k\cdot(1+2l)\). Thus, our algorithm achieves an approximation ratio of \((\Omega(k^2)+k\cdot(1+2l))/(k\cdot(1+2l))=\Omega(k/\log k)\).
\end{proof}
\begin{remark}
It is straightforward to tweak the above analysis to show that for min-max Time Disjoint Walks, \cref{algs:clever} achieves an approximation ratio of \(\Theta(k/\log k)\) on bounded-degree DAGs.
\end{remark}

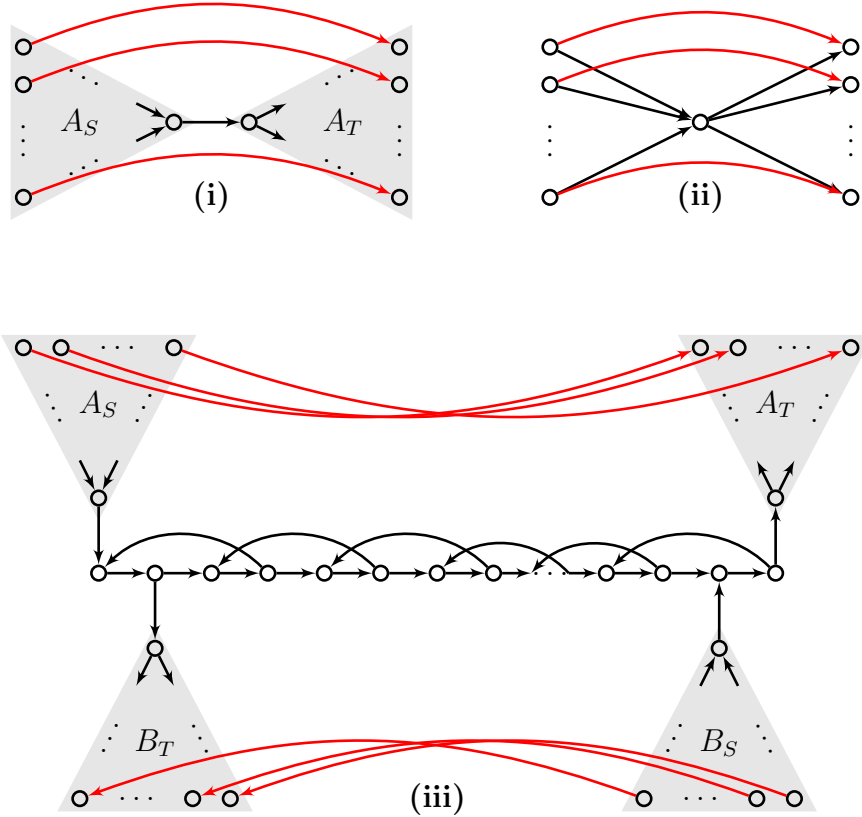
\begin{figure}[t]
\centering
\begin{tikzpicture}[line width=1pt]
\tikzset{vertex/.style = {shape=circle,draw,minimum size=0.5em, inner sep=0pt}}
\tikzset{edge/.style = {->,> = latex'}}

\fill[gray!20] (.83,2.3) -- (.83,-.3) -- (3.3,1); 
\fill[gray!20] (6.17,2.3) -- (6.17,-.3) -- (3.7,1); 

\node[vertex] (1bl) at (1,0) {};
\node[vertex] (1ml) at (1,1.5) {};
\node[vertex] (1tl) at (1,2) {};
\coordinate (1tinvisil) at (2.5, 1.25) {};
\coordinate (1binvisil) at (2.5, 0.75) {};
\path (1bl) -- (1ml) node [font=\large, midway, sloped] {\(\dots\)};
\path (1tl) -- (1tinvisil) node [font=\large, midway, sloped] {\(\dots\)};
\path (1bl) -- (1binvisil) node [font=\large, midway, sloped] {\(\dots\)};
\node[vertex,draw=none] (1AS_label) at (1.75,1) {\large \(A_S\)};

\node[vertex] (1lc) at (3,1) {};
\node[vertex] (1rc) at (4,1) {};

\node[vertex] (1br) at (6,0) {};
\node[vertex] (1mr) at (6,1.5) {};
\node[vertex] (1tr) at (6,2) {};
\coordinate (1tinvisir) at (4.5, 1.25) {};
\coordinate (1binvisir) at (4.5, 0.75) {};
\path (1br) -- (1mr) node [font=\large, midway, sloped] {\(\dots\)};
\path (1tr) -- (1tinvisir) node [font=\large, midway, sloped] {\(\dots\)};
\path (1br) -- (1binvisir) node [font=\large, midway, sloped] {\(\dots\)};
\node[vertex,draw=none] (1AT_label) at (5.25,1) {\large \(A_T\)};

\draw[edge] (1tinvisil) to (1lc);
\draw[edge] (1binvisil) to (1lc);
\draw[edge] (1lc) to (1rc);
\draw[edge] (1rc) to (1tinvisir);
\draw[edge] (1rc) to (1binvisir);

\draw[edge, bend left=22, color=red] (1tl) to (1tr);
\draw[edge, bend left=22, color=red] (1ml) to (1mr);
\draw[edge, bend left=22, color=red] (1bl) to (1br);

\node[vertex,draw=none] (fig1_label) at (3.5,0) {\large \textbf{(i)}};

\node[vertex] (2bl) at (8,0) {};
\node[vertex] (2ml) at (8,1.5) {};
\node[vertex] (2tl) at (8,2) {};
\path (2bl) -- (2ml) node [font=\large, midway, sloped] {\(\dots\)};

\node[vertex] (2c) at (10,1) {};

\node[vertex] (2br) at (12,0) {};
\node[vertex] (2mr) at (12,1.5) {};
\node[vertex] (2tr) at (12,2) {};
\path (2br) -- (2mr) node [font=\large, midway, sloped] {\(\dots\)};

\draw[edge] (2bl) to (2c);
\draw[edge] (2ml) to (2c);
\draw[edge] (2tl) to (2c);

\draw[edge] (2c) to (2br);
\draw[edge] (2c) to (2mr);
\draw[edge] (2c) to (2tr);

\draw[edge, bend left=22, color=red] (2tl) to (2tr);
\draw[edge, bend left=22, color=red] (2ml) to (2mr);
\draw[edge, bend left=22, color=red] (2bl) to (2br);

\node[vertex,draw=none] (fig2_label) at (10,0) {\large \textbf{(ii)}};

\node[vertex] (c1) at (2,-5) {};
\node[vertex] (c2) at (2.75,-5) {};
\node[vertex] (c3) at (3.5,-5) {};
\node[vertex] (c4) at (4.25,-5) {};
\node[vertex] (c5) at (5,-5) {};
\node[vertex] (c6) at (5.75,-5) {};
\node[vertex] (c7) at (6.5,-5) {};
\node[vertex] (c8) at (7.25,-5) {};
\coordinate (c8p5) at (7.75,-5) {};
\coordinate (c9p5) at (8.25,-5) {};
\node[vertex] (c10) at (8.75,-5) {};
\node[vertex] (c11) at (9.5,-5) {};
\node[vertex] (c12) at (10.25,-5) {};
\node[vertex] (c13) at (11,-5) {};
\path (c8p5) -- (c9p5) node [font=\large, midway, sloped] {\(\dots\)};
\draw[edge] (c1) to (c2);
\draw[edge] (c2) to (c3);
\draw[edge] (c3) to (c4);
\draw[edge] (c4) to (c5);
\draw[edge] (c5) to (c6);
\draw[edge] (c6) to (c7);
\draw[edge] (c7) to (c8);
\draw[edge] (c8) to (c8p5);
\draw[edge] (c9p5) to (c10);
\draw[edge] (c10) to (c11);
\draw[edge] (c11) to (c12);
\draw[edge] (c12) to (c13);

\draw[edge, bend right=45] (c13) to (c10);
\draw[edge, bend right=45] (c11) to (c8p5);
\draw[edge, bend right=45] (c9p5) to (c7);
\draw[edge, bend right=45] (c8) to (c5);
\draw[edge, bend right=45] (c6) to (c3);
\draw[edge, bend right=45] (c4) to (c1);

\fill[gray!20] (0.7, -1.83 ) -- (3.3,-1.83) -- (2,-4.3); 
\node[vertex,draw=none] (3AS_label) at (2,-2.75) {\large \(A_S\)};
\node[vertex] (luAS) at (1,-2) {};
\node[vertex] (lpAS) at (1.5,-2) {};
\node[vertex] (lfAS) at (3,-2) {};
\node[vertex] (rAS) at (2,-4) {};
\coordinate (linvisiAS) at (1.75, -3.5) {};
\coordinate (rinvisiAS) at (2.25, -3.5) {};
\path (lpAS) -- (lfAS) node [font=\large, midway, sloped] {\(\dots\)};
\path (luAS) -- (linvisiAS) node [font=\large, midway, sloped] {\(\dots\)};
\path (lfAS) -- (rinvisiAS) node [font=\large, midway, sloped] {\(\dots\)};
\draw[edge] (linvisiAS) to (rAS);
\draw[edge] (rinvisiAS) to (rAS);
\draw[edge] (rAS) to (c1);

\fill[gray!20] (9.7, -1.83 ) -- (12.3,-1.83) -- (11,-4.3); 
\node[vertex,draw=none] (3AT_label) at (11,-2.75) {\large \(A_T\)};
\node[vertex] (luAT) at (10, -2) {};
\node[vertex] (lpAT) at (10.5, -2) {};
\node[vertex] (lfAT) at (12, -2) {};
\node[vertex] (rAT) at (11, -4) {};
\coordinate (linvisiAT) at (10.75, -3.5) {};
\coordinate (rinvisiAT) at (11.25, -3.5) {};
\path (lpAT) -- (lfAT) node [font=\large, midway, sloped] {\(\dots\)};
\path (luAT) -- (linvisiAT) node [font=\large, midway, sloped] {\(\dots\)};
\path (lfAT) -- (rinvisiAT) node [font=\large, midway, sloped] {\(\dots\)};
\draw[edge] (rAT) to (linvisiAT);
\draw[edge] (rAT) to (rinvisiAT);
\draw[edge] (c13) to (rAT);

\draw[edge, bend right=20, color=red] (luAS) to (luAT);
\draw[edge, bend right=20, color=red] (lpAS) to (lpAT);
\draw[edge, bend right=20, color=red] (lfAS) to (lfAT);

\fill[gray!20] (1.45, -8.17 ) -- (4.05,-8.17) -- (2.75,-5.7); 
\node[vertex,draw=none] (3BT_label) at (2.75,-7.25) {\large \(B_T\)};
\node[vertex] (rBT) at (2.75,-6) {};
\node[vertex] (luBT) at (3.75,-8) {};
\node[vertex] (lpBT) at (3.25,-8) {};
\node[vertex] (lfBT) at (1.75,-8) {};
\coordinate (rinvisiBT) at (2.5, -6.5) {};
\coordinate (linvisiBT) at (3, -6.5) {};
\path (lpBT) -- (lfBT) node [font=\large, midway, sloped] {\(\dots\)};
\path (luBT) -- (linvisiBT) node [font=\large, midway, sloped] {\(\dots\)};
\path (lfBT) -- (rinvisiBT) node [font=\large, midway, sloped] {\(\dots\)};
\draw[edge] (rBT) to (linvisiBT);
\draw[edge] (rBT) to (rinvisiBT);
\draw[edge] (c2) to (rBT);

\fill[gray!20] (8.95, -8.17 ) -- (11.55,-8.17) -- (10.25,-5.7); 
\node[vertex,draw=none] (3BT_label) at (10.25,-7.25) {\large \(B_S\)};
\node[vertex] (rBS) at (10.25,-6) {};
\node[vertex] (luBS) at (11.25,-8) {};
\node[vertex] (lpBS) at (10.75,-8) {};
\node[vertex] (lfBS) at (9.25,-8) {};
\coordinate (rinvisiBS) at (10, -6.5) {};
\coordinate (linvisiBS) at (10.5, -6.5) {};
\path (lpBS) -- (lfBS) node [font=\large, midway, sloped] {\(\dots\)};
\path (luBS) -- (linvisiBS) node [font=\large, midway, sloped] {\(\dots\)};
\path (lfBS) -- (rinvisiBS) node [font=\large, midway, sloped] {\(\dots\)};
\draw[edge] (linvisiBS) to (rBS);
\draw[edge] (rinvisiBS) to (rBS);
\draw[edge] (rBS) to (c12);

\draw[edge, bend right=20, color=red] (luBS) to (luBT);
\draw[edge, bend right=20, color=red] (lpBS) to (lpBT);
\draw[edge, bend right=20, color=red] (lfBS) to (lfBT);

\node[vertex,draw=none] (fig3_label) at (6.5,-8) {\large \textbf{(iii)}};
\end{tikzpicture}
\caption{(i): A bounded-degree DAG \(G_k\) upon which \cref{algs:clever} achieves an approximation ratio of \(\Omega(k/\log k)\); (ii): A DAG \(G_k\) upon which \cref{algs:clever} achieves an approximation ratio of \(\Omega(k)\); (iii): A bounded-degree digraph \(G_k\) upon which \cref{algs:clever} achieves an approximation ratio of \(\Omega(k)\).}
\label{fig:master-figure}
\end{figure}
\subsection{Analysis on DAGs}
We show that if we no longer require the graph family in \cref{thm:best-ratio} to have bounded degree, our algorithm loses its improved approximation ratio.
\begin{theorem}\label{thm:DAG-ratio}
\cref{algs:clever} has an approximation ratio of \(\Theta(k)\) for min-sum Time Disjoint Walks on DAGs.
\end{theorem}
\begin{proof}
By \cref{prop:very-easy-bound}, it suffices to construct a family of TDW instances on DAGs, defined over all \(k\in\mathbb{N}_{\geq1}\), for which our algorithm achieves an approximation ratio of \(\Omega(k)\). Construct \(G_k\) by fixing a ``root'' vertex and directly attaching \(2k\) leaves. Orient half of these arcs towards the root, and half of the arcs away from the root. Call each vertex with out-degree 1 a \emph{source}, and each vertex with in-degree 1 a \emph{destination}. Then, arbitrarily pair each source with a unique destination. For each pair, add an arc from the source to the destination (called a ``bypass arc''), and add a demand to \(\mathcal{T}_k\). Finally, let \(\lambda_k\) assign length 2 to each bypass arc, and length 1 to all other arcs. We refer the reader to \cref{fig:master-figure}, part (ii).

We may assume our algorithm does not satisfy demands using the bypass arcs (as all demand-satisfying paths have length 2, and no tie-breaking scheme is specified). Thus, each demand-satisfying path uses the root vertex, which incurs a total delay of \(0+1+\ldots+(k-1)=\Omega(k^2)\) and total path length of \(2k\). Had the bypass arcs been used, no delay would have been required, and the total path length would have still been \(2k\). Thus, our algorithm achieves an approximation ratio of \((\Omega(k^2)+2k)/(2k)=\Omega(k)\).
\end{proof}
\begin{remark}
It is straightforward to slightly tweak the above analysis to show that for min-max Time Disjoint Walks, \cref{algs:clever} has an approximation ratio of \(\Theta(k)\) on DAGs.
\end{remark}

\subsection{Analysis on bounded-degree digraphs}
In this section, we show that if we no longer require the graph family in \cref{thm:best-ratio} to be acyclic, our algorithm loses its improved approximation ratio.
\begin{theorem}\label{thm:bounded-deg-digraphs}
\cref{algs:clever} has an approximation ratio of \(\Theta(k)\) for min-sum Time Disjoint Walks on bounded-degree digraphs.
\end{theorem}
\begin{proof}
By \cref{prop:very-easy-bound}, it suffices to construct a family of TDW instances on bounded-degree digraphs, defined over all \(l\in\N_{\geq2}\) with \(\hat{k}:=2^l, k:=2\hat{k}\), for which our algorithm achieves an approximation ratio of \(\Omega(k)\). Construct \(G_k\) by taking two \((2,l)\)-in-trees \(A_S\) and \(B_S\), and two \((2,l)\)-out-trees \(A_T\) and \(B_T\). Call their roots \(r_{A_S},r_{B_S},r_{A_T}\), and \(r_{B_T}\), respectively. Then, add a ``central path'' \(C\) consisting of vertices \(\{c_1,c_2,\dots,c_{\hat{k}}\}\), ``forward'' arcs \(\{(c_i,c_{i+1})\mid i\in[\hat{k}-1]\}\), and ``backward'' arcs \(\{(c_j,c_{j-3})\mid j\in[4,\hat{k}], j\bmod2=0\}\). Attach the directed trees to the central path with arcs \(\{(r_{A_S},c_1),(r_{B_S},c_{\hat{k}-1}),(c_{\hat{k}},r_{A_T}),(c_2,r_{B_T})\}\). Next, pair each leaf (source) in \(A_S\) with an arbitrary, but unique, leaf (destination) in \(A_T\). Do the same for \(B_S\) and \(B_T\). For each such pair, add an arc from the source to destination (called a ``bypass arc''), and add a demand to \(\mathcal{T}_k\). Finally, let \(\lambda_k\) assign length \(2\hat{k}+2l-1\) to each bypass arc, length \(\hat{k}-1\) to arcs \((r_{B_S},c_{\hat{k}-1})\) and \((c_{\hat{k}},r_{A_T})\), and length \(1\) to all other arcs. We refer the reader to \cref{fig:master-figure}, part (iii).

Observe that for each demand, there exist \emph{two} shortest demand-satisfying paths, each of length \(2\hat{k}+2l-1\). In particular, observe that a demand between leaves of \(A_S\) and \(A_T\) may be satisfied by a bypass arc, or by a path that travels from the source in \(A_S\), towards the root of \(A_S\), onto the central path vertex \(c_1\), along all forward arcs of \(C\), onto the root of \(A_T\), and towards the destination in \(A_T\). Similarly, a demand between leaves of \(B_S\) and \(B_T\) may be satisfied by a bypass arc, or by a path that travels from the source in \(B_S\), towards the root of \(B_S\), onto the central vertex \(c_{\hat{k}-1}\), across \(C\) by alternating between forward and backward arcs (until arriving at \(c_2\)), onto the root of \(B_T\), and towards the destination in \(B_T\). We call the paths that do not use the bypass arcs the ``meandering paths.''

Because our algorithm specifies no tie-breaking scheme, we may assume that it satisfies demands using the meandering paths, and that it alternates between assigning delays to demands from \(A_S\) and assigning delays to demands from \(B_S\) every \emph{four} iterations. In other words, out of the \(2\hat{k}\) demands created above and fed as input to our algorithm, we may assume that those from \(A_S\) to \(A_T\) are labeled with indices \(I_A:=\{i\in[2\hat{k}]\mid \lfloor(i-1)/4\rfloor\equiv 0 \pmod 2\}\), while those from \(B_S\) to \(B_T\) are labeled with \(I_B:=\{i\in[2\hat{k}]\mid \lfloor(i-1)/4\rfloor\equiv 1 \pmod 2\}\).

To understand the suboptimality of this situation, we make several observations that help us determine the values our algorithm assigns to each \(d_i\). First, note that for each \(i\in[2\hat{k}]\), \(z\in[\hat{k}]\), the length of walk \(W_i\) up to vertex \(c_z\) on the central path is:
\begin{align*}
\lambda_k(W_i,c_z) =
\begin{cases}
l+z,&\textbf{if }i\in I_A\\
l+2\hat{k}-z-2\cdot(z\bmod 2),&\textbf{if }i\in I_B
\end{cases}
\end{align*}
Using this, we see that the details of \cref{algs:clever} give us the following relation, which is defined over \(i\in[k],h\in[i-1]\):
\begin{align*}
\mathsf{bad\_delays}_{i,h} =
\begin{cases}
\{d_h\},&\textbf{if }i,h\in I_A\textbf{ or}\\
&\textbf{if }i,h\in I_B\\
\{d_h-2\hat{k}+2z+2\cdot(z\bmod 2)\mid z\in[\hat{k}]\},&\textbf{if }i\in I_B,h\in I_A\\
\{d_h+2\hat{k}-2z-2\cdot(z\bmod2)\mid z\in[\hat{k}]\},&\textbf{if }i\in I_A,h\in I_B
\end{cases}
\end{align*}
Because our algorithm defines \(\mathsf{bad\_delays}_i:=\bigcup_{h\in[i-1]}\mathsf{bad\_delays}_{i,h}\) and \(d_i:=\min(\mathbb{Z}_{\geq0}\setminus\mathsf{bad\_delays}_i)\), observe that the above relation is in fact a \emph{recurrence} relation. As such, after noting that \(d_1=0\), it is straightforward to use the above relation to show by induction that for all \(i\in[2\hat{k}]\),
\begin{align*}
d_i = i-1 + \lfloor\frac{i-1}{8}\rfloor(2\hat{k}-4).
\end{align*}
Thus, our algorithm incurs a total delay of \(\sum_{i\in[2\hat{k}]}(i-1+\lfloor(i-1)/8\rfloor(2\hat{k}-4))=\Omega(\hat{k}^3)=\Omega(k^3)\) and total walk length of \(2\hat{k}\cdot(2\hat{k}+2l-1))=\Theta(\hat{k}^2)=\Theta(k^2)\). Had the algorithm opted to use the bypass arcs, no delay would have been required, and the total walk length would have been the same. Thus, our algorithm achieves an approximation ratio of \(\Omega(k)\).
\end{proof}
\begin{remark}
It is straightforward to slightly tweak the above analysis to show that for min-max Time Disjoint Walks, \cref{algs:clever} has an approximation ratio of \(\Theta(k)\) on bounded-degree digraphs.
\end{remark}

\section{Conclusions}\label{section:conclusions}
In this paper, we introduce Time Disjoint Walks, a new variant of (shortest) Disjoint Paths that similarly seeks to connect \(k\) demands in a network, but relaxes the disjointness constraint by permitting vertices to be shared across multiple walks, as long as no two walks arrive at the same vertex at the same time. We present and analyze a poly-time exact algorithm for both min-sum and min-max TDW on oriented stars, and show that min-sum TDW becomes NP-hard on a slightly more expressive class of inputs. We show that min-sum Time Disjoint Walks is APX-hard, even for DAGs of max degree three. On the other hand, we provide a natural \(\Theta(k/\log k)\)-approximation algorithm for directed acyclic graphs of bounded degree. Interestingly, we also show that for general digraphs with just one of these two properties, the approximation ratio of our algorithm is bumped up to \(\Theta(k)\).

An interesting future work is to tighten the gap between these inapproximability and approximability results for TDW on bounded-degree DAGs. We conjecture that our approximation algorithm is almost optimal, but that our hardness of approximation result can be strengthened to nearly match our algorithm's approximation ratio of \(\Theta(k/\log k)\). This belief is based on the observation that TDW is a complex problem that involves both routing \emph{and} scheduling, and many problems of the latter variety (of size \(n\)) are NP-hard to approximate within a factor of \(n^{1-\epsilon}\), for any \(\epsilon>0\)~\cite{zuck}. One may also wish to explore similar complexity questions for the many variants of Time Disjoint Walks discussed in \cref{section:TDW}.


\bibliographystyle{alpha}
\bibliography{references}

\end{document}